\documentclass[runningheads]{llncs}
\usepackage{amsmath,amsfonts,amssymb}
\usepackage[pdftex]{graphicx}
  \graphicspath{{./pics/}}
  \DeclareGraphicsExtensions{.pdf,.jpg}
\usepackage{tikz}
  \usetikzlibrary{arrows,automata,shapes}
  \usetikzlibrary{positioning}
\usepackage{todonotes}
  
\newcommand{\eps}{\varepsilon}
\newcommand{\preq}{\preccurlyeq}
\DeclareMathOperator{\alp}{\rm alph}

\begin{document}
\frontmatter
\pagestyle{headings}

\mainmatter
\title{On Upper and Lower Bounds on the Length\\ of Alternating Towers}
\authorrunning{\v{S}. Holub, G. Jir\'{a}skov\'{a}, T. Masopust}
\author{
  \v St\v ep\'an Holub\,\inst{1,}%
  \thanks{Research supported by the Czech Science Foundation grant number 13-01832S.}
  \and
  Galina Jir\'{a}skov\'{a}\,\inst{2,}%
  \thanks{Research supported by grant APVV-0035-10.}
  \and
  Tom\'{a}\v{s} Masopust\,\inst{3,}%
  \thanks{Research supported by RVO~67985840 and by the DFG in grant KR~4381/1-1.}%
  }
\institute{
  Dept. of Algebra, Charles University, Sokolovsk\'a 83, 175 86 Praha, Czech Republic\\
  \email{holub@karlin.mff.cuni.cz}
  \and
  Mathematical Institute, Slovak Academy of Sciences\\
  Gre{\v s}{\' a}kova 6, 040 01 Ko\v{s}ice, Slovak Republic\\
  \email{jiraskov@saske.sk}
  \and
  Institute of Mathematics, ASCR,
  {\v Z}i{\v z}kova 22, 616 62 Brno, Czech Republic, and TU Dresden, Germany\\
  \email{masopust@math.cas.cz}\\
}

\maketitle
\begin{abstract}
  A tower between two regular languages is a sequence of strings such that all strings on odd positions belong to one of the languages, all strings on even positions belong to the other language, and each string can be embedded into the next string in the sequence. It is known that if there are towers of any length, then there also exists an infinite tower. We investigate upper and lower bounds on the length of finite towers between two regular languages with respect to the size of the automata representing the languages in the case there is no infinite tower. This problem is relevant to the separation problem of regular languages by piecewise testable languages.
\end{abstract}

\section{Introduction}
  The separation problem appears in many disciplines of mathematics and computer science, such as algebra and logic~\cite{mfcsPlaceRZ13,PlaceZ14}, or databases and query answering~\cite{icalp2013}. Given two languages $K$ and $L$ and a family of languages $\mathcal{F}$, the problem asks whether there exists a language $S$ in $\mathcal{F}$ such that $S$ includes one of the languages $K$ and $L$, and it is disjoint with the other. Recently, it has been independently shown in~\cite{icalp2013} and~\cite{mfcsPlaceRZ13} that the separation problem for two regular languages given as NFAs and the family of piecewise testable languages is decidable in polynomial time with respect to both the number of states and the size of the alphabet. It should be noted that an algorithm polynomial with respect to the number of states and exponential with respect to the size of the alphabet has been known in the literature~\cite{Almeida-jpaa90,AlmeidaZ-ita97}. In~\cite{icalp2013}, the separation problem has been shown to be equivalent to the non-existence of an infinite tower between the input languages. Namely, the languages have been shown separable by a piecewise testable language if and only if there does not exist an infinite tower. In~\cite{mfcsPlaceRZ13}, another technique has been used to prove the polynomial time bound for the decision procedure, and a doubly exponential upper bound on the index of the separating piecewise testable language has been given. This information can then be further used to construct a separating piecewise testable language.
  
  However, there exists a simple (in the meaning of description, not complexity) method to decide the separation problem and to compute the separating piecewise testable language, whose running time depends on the length of the longest finite tower. The method is recalled in Section~\ref{secAlg}. This observation has motivated the study of this paper to investigate the upper bound on the length of finite towers in the presence of no infinite tower. So far, to the best of our knowledge, the only published result in this direction is a paper by Stern~\cite{Stern-tcs85}, who has given an exponential upper bound $2^{{|\Sigma|}^2 N}$ on the length of the tower between a piecewise testable language and its complement, where $N$ is the number of states of the minimal deterministic automaton. 
  
  Our contribution in this paper are upper and lower bounds on the length of maximal finite towers between two regular languages in the case no infinite towers exist. These bounds depend on the size of the input (nondeterministic) automata. The upper bound is exponential with respect to the size of the input alphabet. More precisely, it is polynomial with respect to the number of states with the cardinality of the input alphabet in the exponent (Theorem~\ref{thm03}). Concerning the lower bounds, we show that the bound is tight for binary languages up to a linear factor (Theorem~\ref{thm:quadratic}), that a cubic tower with respect to the number of states exists (Theorem~\ref{thm:cubic}), and that an exponential lower bound with respect to the size of the input alphabet can be achieved (Theorem~\ref{thm:exp}).

\section{Preliminaries}
  We assume that the reader is familiar with automata and formal language theory. The cardinality of a set $A$ is denoted by $|A|$ and the power set of $A$ by $2^A$. An alphabet $\Sigma$ is a finite nonempty set. The free monoid generated by $\Sigma$ is denoted by $\Sigma^*$. A string over $\Sigma$ is any element of $\Sigma^*$; the empty string is denoted by $\eps$. For a string $w\in\Sigma^*$, $\alp(w)\subseteq\Sigma$ denotes the set of all letters occurring in $w$.

  We define {\em (alternating subsequence) towers\/} as a generalization of Stern's alternating towers~\cite{Stern-tcs85}. For strings $v = a_1 a_2 \cdots a_n$ and $w \in \Sigma^* a_1 \Sigma^* \cdots \Sigma^* a_n \Sigma^*$, we say that $v$ is a {\em subsequence\/} of $w$ or that $v$ can be {\em embedded\/} into $w$, denoted by $v \preccurlyeq w$. For languages $K$ and $L$ and the subsequence relation $\preccurlyeq$, we say that a sequence $(w_i)_{i=1}^{k}$ of strings is an {\em (alternating subsequence) tower between $K$ and $L$\/} if $w_1 \in K \cup L$ and, for all $i = 1, \ldots, k-1$,
  \begin{itemize}
    \item $w_i \preccurlyeq w_{i+1}$,
    \item $w_i \in K$ implies $w_{i+1} \in L$, and
    \item $w_i \in L$ implies $w_{i+1} \in K$.
  \end{itemize}
  
  We say that $k$ is the {\em length\/} of the tower. Similarly, we define an infinite sequence of strings to be an {\em infinite (alternating subsequence) tower between $K$ and $L$}. If the languages are clear from the context, we omit them. Notice that the languages are not required to be disjoint, however, if there exists a $w \in K \cap L$, then there exists an infinite tower, namely $w, w, w, \ldots$.
  
  For two languages $K$ and $L$, we say that the {\em language $K$ can be embedded into the language $L$}, denoted $K\preccurlyeq L$, if for each string $w$ in $K$, there exists a string $w'$ in $L$ such that $w\preccurlyeq w'$. We say that a {\em string $w$ can be embedded into the language $L$}, denoted $w\preccurlyeq L$, if $\{w\}\preccurlyeq L$.

  A {\em nondeterministic finite automaton\/} (NFA) is a 5-tuple $M = (Q,\Sigma,\delta,Q_0,F)$, where $Q$ is the finite nonempty set of states, $\Sigma$ is the input alphabet, $Q_0\subseteq Q$ is the set of initial states, $F\subseteq Q$ is the set of accepting states, and $\delta:Q\times\Sigma\to 2^Q$ is the transition function that can be extended to the domain $2^Q\times\Sigma^*$. The language {\em accepted\/} by $M$ is the set $L(M) = \{w\in\Sigma^* \mid \delta(Q_0, w) \cap F \neq\emptyset\}$. A {\em path\/} $\pi$ is a sequence of states and input symbols $q_0, a_0, q_1, a_1, \ldots, q_{n-1}, a_{n-1}, q_n$, for some $n\ge 0$, such that $q_{i+1} \in \delta(q_i,a_i)$, for all $i=0,1,\ldots,n-1$. The path $\pi$ is {\em accepting\/} if $q_0\in Q_0$ and $q_n\in F$. We also use the notation $q_0 \xrightarrow{a_1a_2\cdots a_{n-1}} q_{n}$ to denote a path from $q_0$ to $q_n$ under a string $a_1a_2\cdots a_{n-1}$.

  The NFA $M$ has a {\em cycle over an alphabet $\Gamma\subseteq\Sigma$\/} if there exists a state $q$ and a string $w$ over $\Sigma$ such that $\alp(w)=\Gamma$ and $q\xrightarrow{w} q$.
  
  We assume that there are no useless states in the automata under consideration, that is, every state appears on an accepting path.

\section{Computing a Piecewise Testable Separator\,\protect\footnote{The method recalled here is not the original work of this paper and the credit for this should go to the authors of~\cite{pc2013}, namely to Wim Martens and Wojciech Czerwi\'nski.}
}\label{secAlg}
  We now motivate our study by recalling a ``simple'' method~\cite{pc2013} solving the separation problem of regular languages by piecewise testable languages and computing a piecewise testable separator, if it exists. Our motivation to study the length of towers comes from the fact that the running time of this method depends on the maximal length of finite towers.
  
  Let $K$ and $L$ be two languages. A language $S$ \emph{separates $K$ from $L$\/} if $S$ contains $K$ and does not intersect $L$. Languages $K$ and $L$ are \emph{separable by a family $\mathcal{F}$\/} if there exists a language $S$ in $\mathcal{F}$ that separates $K$ from $L$ or $L$ from $K$.

  A regular language is {\em piecewise testable\/} if it is a finite boolean combination of languages of the form $\Sigma^* a_1 \Sigma^* a_2 \Sigma^* \cdots \Sigma^* a_k \Sigma^*$, where $k\ge 0$ and $a_i\in \Sigma$, see~\cite{Simon1972,Simon1975} for more details.

  Given two disjoint regular languages $L_0$ and $R_0$ represented as NFAs. We construct a decreasing sequence of languages $\ldots\preccurlyeq R_2 \preccurlyeq L_2 \preccurlyeq R_1 \preccurlyeq L_1 \preccurlyeq R_0$ as follows, show that a separator exists if and only if from some point on all the languages are empty, and use them to construct a piecewise testable separator.
  
  For $k\ge 1$, let $L_k=\{ w \in L_{k-1} \mid w \preq R_{k-1}\}$ be the set of all strings of $L_{k-1}$ that can be embedded into $R_{k-1}$, and let $R_k=\{ w\in R_{k-1} \mid w\preq L_k\}$, see Fig.~\ref{fig2}.
  \begin{figure}
    \centering
    \begin{tikzpicture}
      \draw (1,2) node {$L_0$};
      \draw (5,2) node {$R_0$};
      \draw (1,1.5) node {$w_1\in L_1$};
      \draw (5,1) node {$R_1$};
      \draw (1,0.5) node {$w_2\in L_2$};
      \draw (5,0) node {$R_2$};
      \draw (3,0) node {$\vdots$};
      \draw[->] (1.6,1.6) -- (4.8,2);
      \draw[->] (4.8,1.1) -- (1.6,1.5);
      \draw[->] (1.6,.6) -- (4.8,1);
      \draw[->] (4.8,.05) -- (1.6,.5);
    \end{tikzpicture}
    \caption{The sequence of languages; an arrow stands for the embedding relation $\preccurlyeq$.}
    \label{fig2}
  \end{figure}
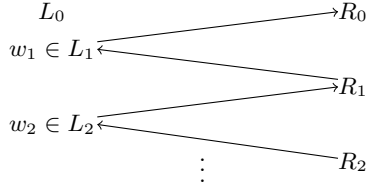
  Let $K$ be a language accepted by an NFA $A=(Q,\Sigma,\delta,Q_0,F)$, and let $\eps(K)$ denote the language accepted by the NFA $A_\eps=(Q,\Sigma,\delta_\eps,Q_0,F)$, where $\delta_\eps(q,a)=\delta(q,a)$ and $\delta_\eps(q,\eps)=\bigcup_{a\in\Sigma} \delta(q,a)$. Then $L_k=L_{k-1}\cap \eps(R_{k-1})$ (analogously for $R_k$), hence the languages are regular.
  
  We now show that there exists a constant $B\ge 1$ such that $L_{B}=L_{B+1}=\ldots$, which also implies $R_B=R_{B+1}=\ldots$. Assume that no such constant exists. Then there are infinitely many strings $w_{\ell}\in L_{\ell}\setminus L_{\ell+1}$, for all $\ell\ge 1$, as depicted in Fig.~\ref{fig2}. By Higman's lemma~\cite{higman}, there exist $i < j$ such that $w_i\preccurlyeq w_j$, hence $w_i \preq R_{j-1}$, which is a contradiction because $w_i \not\preq R_{i}$ and $ R_{j-1}\subseteq R_i$. 

  By construction, languages $L_B$ and $R_B$ are mutually embeddable into each other, $L_B\preq R_B \preq L_B$, which describes a way how to construct an infinite tower. Thus, if there is no infinite tower, languages $L_B$ and $R_B$ must be empty. 

  The constant $B$ depends on the length of the longest finite tower. Let $(w_i)_{i=1}^{r}$ be a maximal finite tower between $L_0$ and $R_0$ and assume that $w_r$ belongs to $L_0$. In the first step, the method eliminates all strings that cannot be embedded into $R_0$, hence $w_r$ does not belong to $L_1$, but $(w_i)_{i=1}^{r-1}$ is a tower between $L_1$ and $R_0$. Thus, in each step of the algorithm, all maximal strings of all finite towers (belonging to the language under consideration) are eliminated, while the rests of towers still form towers between the resulting languages. Therefore, as long as there is a maximal finite tower, the algorithm can make another step.
  
  Assume that there is no infinite tower ($L_B=R_B=\emptyset$). We use the languages computed above to construct a piecewise testable separator. For a string $w=a_1a_2\cdots a_\ell$, we define $L_w = \Sigma^* a_1 \Sigma^*a_2\Sigma^* \cdots \Sigma^* a_\ell \Sigma^*$, which is piecewise testable by definition. Let $up(L) = \bigcup_{w\in L} L_w$. The language $up(L)$ is regular and its NFA is constructed from an NFA for $L$ by adding self-loops under all letters to all states, see~\cite{pkps2014} for more details. By Higman's Lemma~\cite{higman}, $up(L)$ can be written as a finite union of languages of the form $L_w$, for some $w\in L$, hence it is piecewise testable. For $k=B,B-1,\ldots,1$, we define the piecewise testable languages $ S_k = up(R_0\setminus R_k) \setminus up(L_0\setminus L_k) $ and show that $S = \bigcup_{k=1}^{B} S_k$ is a piecewise testable separator of $L_0$ and $R_0$. 
  
  To this end, we show that $L_0\cap S_k=\emptyset$ and $R_0\subseteq S$. To prove the former, let $w\in L_0$. If $w\in L_0 \setminus L_k$, then $w\in up(L_0 \setminus L_k)$, hence $w\notin S_k$. If $w\in L_k$ and $w\in up(R_0\setminus R_k)$, then there is $v\in R_0\setminus R_k$ such that $v\preccurlyeq w$. However, $R_k=\{ u\in R_0 \mid u\preq L_k\}$, hence $v\in R_k$, a contradiction. Thus $L_0\cap S_k=\emptyset$. To prove the later, we show that $R_{k-1}\setminus R_k\subseteq S_k$. Then $R_0= \bigcup_{k=1}^{B} (R_{k-1}\setminus R_k) \subseteq S$. To show this, we have $R_{k-1}\setminus R_k\subseteq R_0\setminus R_k \subseteq up(R_0\setminus R_k)$. If $w\in R_{k-1}$ and $w\in up(L_0\setminus L_k)$, then there is $v\in L_0\setminus L_k$ such that $v\preccurlyeq w$. However, $L_k=\{u \in L_0 \mid u\preq R_{k-1}\}$, hence $v\in L_k$, a contradiction. Thus, we have shown that $L_0\cap S=\emptyset$ and $R_0\subseteq S$. Moreover, $S$ is piecewise testable because it is a finite boolean combination of piecewise testable languages.

\section{The Length of Towers}
  Recall that it was shown in~\cite{icalp2013} that there is either an infinite tower or a constant bound on the length of any tower. We now establish an upper bound on the length of finite towers.

  \begin{theorem}\label{thm03}
    Let $A_0$ and $A_1$ be NFAs with at most $n$ states over an alphabet $\Sigma$ of cardinality $m$, and assume that there is no infinite tower between the languages $L(A_0)$ and $L(A_1)$. Let $(w_i)_{i=1}^r$ be a tower between $L(A_0)$ and $L(A_1)$ such that $w_i\in L(A_{i\bmod 2})$. Then $r \le \frac{n^{m+1}-1}{n-1}$.
  \end{theorem}
  \begin{proof}
    First, we define some new concepts. We say that $w=v_1v_2\cdots v_k$ is a \emph{cyclic factorization} of $w$ with respect to a pair of states $(q,q')$ in an automaton~$A$, if there is a sequence of states 
    $q_0,\dots,q_{k-1},q_{k}$ such that
    $q_0=q$, $q_{k}=q'$, and
    $ q_{i-1} \stackrel {v_i} \longrightarrow q_{i}$,
    for each $i=1,2,\dots k$, and either $v_i$ is a letter,
    or the path $q_{i-1} \stackrel {v_i} \longrightarrow q_{i}$ contains a cycle over $\alp(v_i)$.
    We call $v_i$ a \emph{letter factor} if it is a letter and $q_{i-1}\neq q_i$,
    and a \emph{cycle factor} otherwise. The factorization is \emph{trivial} if $k=1$. 
    Note that this factorization is closely related to the one given in~\cite{Almeida-jpaa90}, see also \cite[Theorem~8.1.11]{AlmeidaBook}.

    We first show that if $q'\in \delta(q, w)$ in some automaton $A$ with $n$  states, then $w$ has a cyclic factorization $v_1v_2\cdots v_k$ with respect to $(q,q')$ that contains at most $n$ cycle factors and at most $n-1$ letter factors. Moreover, if $w$ does not admit the trivial factorization with respect to $(q,q')$, then $\alp(v_i)$ is a strict subset of $\alp(w)$ for each cycle factor $v_i$, $i=1,2,\dots,k$.

 Consider a  path $\pi$ of the automaton $A$
 from $q$ to $q'$ labeled by a string~$w$.
 Let $q_0=q$. Define the factorization  $w=v_1 v_2 \cdots v_k$ inductively
 by the following greedy strategy. 
 Assume we have defined the factors $v_1,  v_2\ldots, v_{i-1}$
 such that $w = v_1 \cdots v_{i-1} w'$
 and $q_0 \xrightarrow{ v_1 v_2\cdots v_{i-1}} q_{i-1}$.
 The factor $v_i$ is defined as the label
 of the longest possible initial segment $\pi_i$ 
 of the path $q_{i-1}\xrightarrow{w'} q'$ 
 such that either $\pi_i$ contains a cycle over $\alp(v_i)$
 or
 $\pi_i=q_{i-1},a,q_{i}$, where $v_i=a$, so $v_i$ is a letter.
 Such a  factorization is well defined,
 and it is a cyclic factorization of $w$.

 Let $p_i$, $i=1,\dots,k$, be a state
 such that the path $q_{i-1} \stackrel {v_i} \longrightarrow q_{i}$
 contains a cycle $p_i\rightarrow p_i$ over $\alp(v_i)$ if $v_i$ is a cycle factor,
 and $p_i=q_{i-1}$ if $v_i$ is a letter factor.
    If $p_i=p_j$ with $i<j$ such that $v_i$ and $v_j$ are cycle factors,
 then we have a contradiction with the maximality of $v_i$ since
 $q_{i-1} \xrightarrow{v_i v_{i+1}\cdots v_j}  q_{j}$
 contains  a cycle $p_i\rightarrow p_i$ from $p_i$ to $p_i$
 over the alphabet $\alp(v_i v_{i+1}\cdots v_j)$.
 Therefore the factorization contains at most $n$ cycle factors.

 Note that $v_i$ is a letter factor only if the state $p_i$,
 which is equal to $q_{i-1}$ in such a case, 
 has no reappearance in the path $q_{i-1}\xrightarrow{v_i \cdots v_k} q'$.
 This implies that there are at most $n-1$ letter factors. 
 Finally, if $\alp(v_i)=\alp(w)$, then $v_i=v_1=w$
 follows from the maximality of $v_1$.

    We now define inductively cyclic factorizations of $w_i$, such that the factorization of $w_{i-1}$ is a refinement of the factorization of $w_{i}$. Let $w_r=v_{r,1}v_{r,2}\cdots v_{r,k_r}$ be a cyclic factorization of $w_r$ defined, as described above, by some accepting path in the automaton $A_{r\bmod 2}$.
    Factorizations $w_{i-1}=v_{i-1,1}v_{i-1,2}\cdots v_{i-1,k_{i-1}}$ are defined as follows. Let 
    \[
      w_{i-1}=v'_{i,1}v'_{i,2}\cdots v'_{i,k_{i}}\,,
    \]
    where $v'_{i,j}\preq v_{i,j}$, for each $j=1,2,\dots,k_{i}$; note that such a factorization exists since $w_{i-1}\preq w_{i}$. Then $v_{i-1,1}v_{i-1,2}\cdots v_{i-1,k_{i-1}}$ is defined 
as a concatenation of  cyclic factorizations of $v'_{i,j}$, $j=1,2,\dots,k_{i}$, corresponding to an accepting path of $w_{i-1}$ in  $A_{i-1\bmod 2}$. The cyclic factorization of the empty string is defined as empty. Note also that a letter factor 
of $w_{i}$ either disappears in $w_{i-1}$,
 or it is  ``factored'' into a letter factor.

In order to measure the height of a tower, we introduce
a weight function $f$ of factors in a factorization $v_1v_2\cdots v_k$.
First, let
\[
g(x)= n\frac{{n^x}-1}{n-1}\,.
\]
Note that $g$ satisfies $g(x+1)=ng(x)+(n-1)+1$.
Now, let 
$f(v_i)=1$ if $v_i$ is a letter factor, and let $f(v_i)=g(|\alp(v_i)|)$ if $v_{i}$ is a cycle factor.
Note that, by definition, $f(\eps)=0$.
 The weight of the factorization $v_1v_2\cdots v_k$ is then defined by
\[W(v_1v_2\cdots v_k)=\sum_{i=1}^k f(v_i)\,.\]
    Let 
    \[
      W_i=W(v_{i,1}v_{i,2}\cdots v_{i,k_i}).
    \]
    We claim that $W_{i-1}<W_{i}$ for each $i=2,\dots,r$. Let $v_1v_2\cdots v_k$ be the fragment of the cyclic factorization of $w_{i-1}$ that emerged as the cyclic factorization of $v_{i,j}'\preq v_{i,j}$. If the factorization is not trivial, then, by the above analysis, 
        \begin{align*}
      W(v_1v_2\cdots v_k)&\leq n-1 + n\cdot g(|\alp(v_{i,j})|-1)
       <g(|\alp(v_{i,j})|)=f(v_{i,j}). 
    \end{align*}
    Similarly, we have $f(v_{i,j}')<f(v_{i,j})$ if $|\alp(v_{i,j}')|<|\alp(v_{i,j})|$. Altogether, we have $W_{i-1}<W_{i}$ as claimed, unless 
				\begin{itemize}
			\item $k_{i-1}=k_{i}$,
			\item the factor $v_{i-1,j}$ is a letter factor if and only 
if $v_{i,j}$ is a letter factor, and
			\item $\alp(v_{i-1,j})= \alp(v_{i,j})$ for all $j=1,2,\dots,k_i$.
		\end{itemize}
		Assume that such a situation takes place, and show that it leads to an infinite tower.
 Let $L$ be the language of strings $z_1z_2\cdots z_{k_i}$ such that $z_j=v_{i,j}$ if $v_{i,j}$ is a letter factor, and $z_j\in (\alp(v_{i,j}))^*$ if $v_{i,j}$ is a cycle factor. Since $w_i\in L(A_{i\bmod 2})$ and $w_{i-1}\in L(A_{i-1\bmod 2})$ holds, the definition of a cycle factor implies that, for each $z\in L$, there is some $z'\in L(A_0)\cap L$ such that $z\preq z'$, and also  $z''\in L(A_1)\cap L$ such that $z\preq z''$. The existence of an infinite tower follows. We have therefore proved $W_{i-1}<W_{i}$. 

    The proof is completed, since $W_r\leq f(w_r)\leq g(m)$, $W_1\geq 0$, and the bound in the claim is equal to $g(m)+1$.
  \qed\end{proof}

  For binary regular languages, we now show that there exists a tower of length at least $n^2 - O(n)$ between two binary regular languages having no infinite tower and represented by automata with at most $n$ states.

  \begin{theorem}\label{thm02}
  \label{thm:quadratic}
    The upper bound $\frac{n^3-1}{n-1}$
    on the length of a maximal tower is tight for binary languages up to a linear factor.
  \end{theorem}
  \begin{proof}
    Let $n$ be an odd number and define the automata $A_0$ and $A_1$ with $n-1$ and $n$ states as depicted in Figs.~\ref{ex01_1} and~\ref{ex01_2}, respectively.
    
    The automaton $A_0=(\{1,2,\ldots,n-1\},\{a,b\},\delta_0,1,\{n-1\})$ consists 
    of an $a$-path from state 1 through states $2,3,\ldots,n-3$, respectively, to state $n-2$, 
    of $a$-transitions from state 1 to all states but itself and the final state,
    of self-loops under $b$ in all but the states $n-2$ and $n-1$,
    and of a $b$-cycle from $n-2$ to $n-1$ and back to $n-2$.
    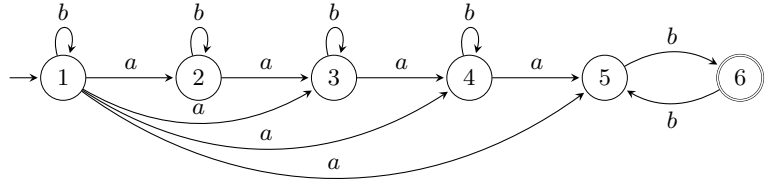
\begin{figure}[t]
      \centering
      \begin{tikzpicture}[->,>=stealth,shorten >=1pt,auto,node distance=1.8cm,
        state/.style={ellipse,minimum size=6mm,very thin,draw=black,initial text=}]
        \node[state,initial]    (1) {$1$};
        \node[state]            (2) [right of=1] {$2$};
        \node[state]            (3) [right of=2] {$3$};
        \node[state]            (4) [right of=3] {$4$};
        \node[state]            (5) [right of=4] {$5$};
        \node[state,accepting]  (6) [right of=5] {$6$};
        \path
          (1) edge node {$a$} (2)
          (1) edge[bend right=30] node {$a$} (3)
          (1) edge[bend right=33] node {$a$} (4)
          (1) edge[bend right=36] node {$a$} (5)
          (2) edge node {$a$} (3)
          (3) edge node {$a$} (4)
          (4) edge node {$a$} (5)
          (1) edge[loop above] node {$b$} (1)
          (2) edge[loop above] node {$b$} (2)
          (3) edge[loop above] node {$b$} (3)
          (4) edge[loop above] node {$b$} (4)
          (5) edge[bend left] node {$b$} (6)
          (6) edge[bend left] node {$b$} (5);
      \end{tikzpicture}
      \caption{Automaton $A_0$; $n-1=6$.}
      \label{ex01_1}
    \end{figure}
    
    The automaton $A_1=(\{1,2,\ldots,n\},\{a,b\},\delta_1,1,\{1,n\})$ consists
    of a $b$-path from state 1 through states $2,3,\ldots,n-1$, respectively, to state $n$, 
    of an $a$-transition from state $n$ to state 1,
    and of $b$-transitions going from state 1 to all even-labeled states.
    
    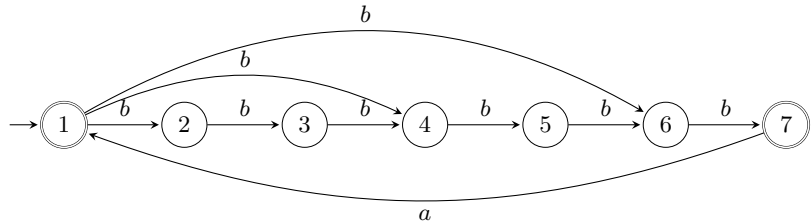
\begin{figure}[b]
      \centering
      \begin{tikzpicture}[->,>=stealth,shorten >=1pt,auto,node distance=1.6cm,
        state/.style={ellipse,minimum size=6mm,very thin,draw=black,initial text=}]
        \node[state,initial,accepting]  (1) {$1$};
        \node[state]                    (2) [right of=1] {$2$};
        \node[state]                    (3) [right of=2] {$3$};
        \node[state]                    (4) [right of=3] {$4$};
        \node[state]                    (5) [right of=4] {$5$};
        \node[state]                    (6) [right of=5] {$6$};
        \node[state,accepting]          (7) [right of=6] {$7$};
        \path
          (1) edge node {$b$} (2)
          (1) edge[bend left=25] node {$b$} (4)
          (1) edge[bend left=30] node {$b$} (6)
          (2) edge node {$b$} (3)
          (3) edge node {$b$} (4)
          (4) edge node {$b$} (5)
          (5) edge node {$b$} (6)
          (6) edge node {$b$} (7)
          (7) edge[bend left=20] node {$a$} (1);
      \end{tikzpicture}
      \caption{Automaton $A_1$; $n=7$.}
      \label{ex01_2}
    \end{figure}
    Consider the string
    \[
      (b^{n-1} a)^{n-3} (b^{n-1} b)\,.
    \]
    This string consists of $n-2$ parts of length $n$ and belongs to $L(A_0)$. Note that deleting the last letter $b$ results in a string that belongs to $L(A_1)$. Deleting another letter $b$ from the right results in a string belonging again to the language $L(A_0)$. We can continue in this way alternating between the languages until the letter $a$ is the last letter, 
    that is, until the string $(b^{n-1} a)^{n-3}$,
    which belongs to $L(A_1)$. Now, we delete the last two letters, namely the string $ba$, which results in a string from $L(A_0)$, and we can continue with deleting the last letters $b$ again as described above. Moreover, we cannot accept the prefix $b^{n-2}$ in $A_0$, hence the length of the tower is at least $n(n-2) - (n-3) - (n-2) = n^2 - 4n + 5$. 
    
    To show that there is no infinite tower between the languages $L(A_0)$ and $L(A_1)$, we can use the techniques described in~\cite{icalp2013,mfcsPlaceRZ13}, or to use the algorithm presented in Section~\ref{secAlg}. 
    We can also notice that letter $a$ can appear at most $n-3$ times in any string from $L(A_0)$ and that after at most $n-1$ occurrences of letter $b$, letter $a$ must appear in a string from $L(A_1)$. As the languages are disjoint, any infinite tower would have to contain a string from $L(A_1)$ of length more than $n \cdot (n-3) + (n-1)$. But any such string in $L(A_1)$ must contain at least $n-2$ occurrences of letter $a$, hence it cannot be embedded into any string of $L(A_0)$. This means that there cannot be an infinite tower.
  \qed\end{proof}

  In Theorem~\ref{thm:quadratic}, we have shown that there exists a tower of a quadratic length between two binary languages having no infinite tower. Now we show that there exist two quaternary languages having a tower of length more than quadratic.

\begin{theorem}\label{thm04}
\label{thm:cubic}
  There exist two languages with no infinite tower having a finite tower of a cubic length.
\end{theorem}
\begin{proof}
  Let $n$ be a number divisible by four and define the automata $A_0$ and $A_1$ with $n-1$ and $n$ states as shown in Figs.~\ref{ex02_1} and~\ref{ex02_2}, respectively.
  
  The automaton $A_0=(\{1,2,\ldots,n-1\},\{a,b,c,d\},\delta_0,1,\{n-1\})$ consists of an $a$-path through states $1,2,\ldots,n-2$, respectively, of $a$-transitions from state 1 to all other states but itself and the final state, of self-loops under symbols $b,c,d$ in all but the final state, and of a $b$-transition from all, but the final state, to the final state.
  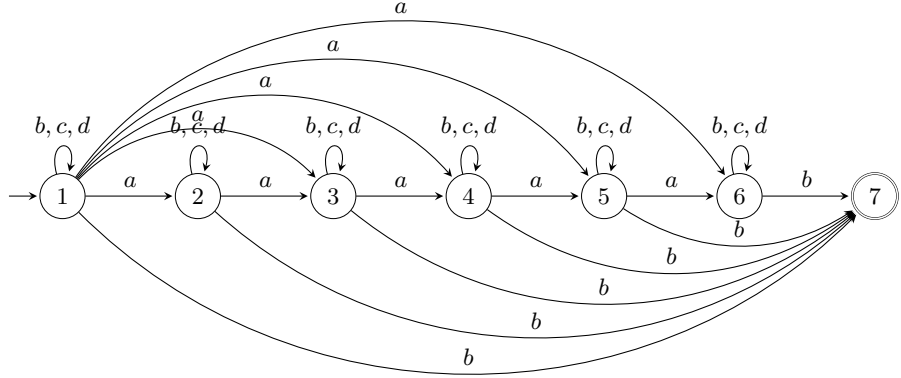
\begin{figure}[t]
    \centering
    \begin{tikzpicture}[->,>=stealth,shorten >=1pt,auto,node distance=1.8cm,
      state/.style={ellipse,minimum size=6mm,very thin,draw=black,initial text=}]
      \node[state,initial]    (1) {$1$};
      \node[state]            (2) [right of=1] {$2$};
      \node[state]            (3) [right of=2] {$3$};
      \node[state]            (4) [right of=3] {$4$};
      \node[state]            (5) [right of=4] {$5$};
      \node[state]            (6) [right of=5] {$6$};
      \node[state,accepting]  (7) [right of=6] {$7$};
      \path
        (1) edge node {$a$} (2)
        (1) edge[bend left=46] node {$a$} (3)
        (1) edge[bend left=49] node {$a$} (4)
        (1) edge[bend left=52] node {$a$} (5)
        (1) edge[bend left=55] node {$a$} (6)
        (1) edge[loop above] node {$b,c,d$} (1)
        (1) edge[bend right=45] node {$b$} (7)
        (2) edge node {$a$} (3)
        (2) edge[loop above] node {$b,c,d$} (2)
        (2) edge[bend right=42] node {$b$} (7)
        (3) edge node {$a$} (4)
        (3) edge[loop above] node {$b,c,d$} (3)
        (3) edge[bend right=39] node {$b$} (7)
        (4) edge node {$a$} (5)
        (4) edge[loop above] node {$b,c,d$} (4)
        (4) edge[bend right=36] node {$b$} (7)
        (5) edge node {$a$} (6)
        (5) edge[loop above] node {$b,c,d$} (5)
        (5) edge[bend right=33] node {$b$} (7)
        (6) edge node {$b$} (7)
        (6) edge[loop above] node {$b,c,d$} (6);
    \end{tikzpicture}
    \caption{Automaton $A_0$; $n-1=7$.}
    \label{ex02_1}
  \end{figure}
  
  The automaton $A_1=(\{1,2,\ldots,n\},\{a,b,c,d\},\delta_1,1,\{\frac{n}{2},n\})$ consists of two parts. 
  The first part is constituted by states $1,2,\ldots,\frac{n}{2}$ with a $d$-path through states $1,2,\ldots,\frac{n}{2}$, respectively, by self-loops under $b,c$ in states $1,2,\ldots,\frac{n}{2}-1$, and by $d$-transitions from state 1 to all of states $2,3,\ldots,\frac{n}{2}$. The second part is constituted by states $\frac{n}{2},\ldots,n$ with a $bc$-path through states $\frac{n}{2},\ldots,n-2$, respectively, by $a$-transitions from state $n-1$ to states $1$ and $n$, by a $c$-transition from state $n-1$ to state $n$, and by $b$-transitions from state $\frac{n}{2}$ to all odd-numbered states between $\frac{n}{2}$ and $n-1$.
  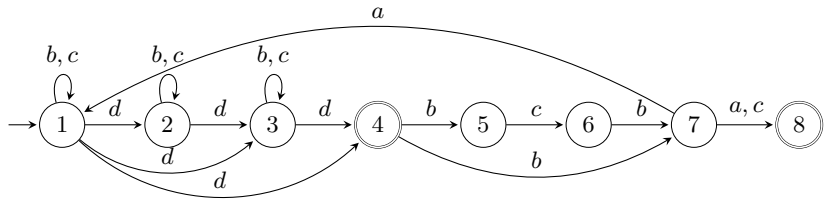
\begin{figure}[b]
    \centering
    \begin{tikzpicture}[->,>=stealth,shorten >=1pt,auto,node distance=1.4cm,
      state/.style={ellipse,minimum size=6mm,very thin,draw=black,initial text=}]
      \node[state,initial]            (1) {$1$};
      \node[state]                    (2) [right of=1] {$2$};
      \node[state]                    (3) [right of=2] {$3$};
      \node[state,accepting]          (4) [right of=3] {$4$};
      \node[state]                    (5) [right of=4] {$5$};
      \node[state]                    (6) [right of=5] {$6$};
      \node[state]                    (7) [right of=6] {$7$};
      \node[state,accepting]          (8) [right of=7] {$8$};
      \path
        (1) edge[loop above] node {$b,c$} (1)
        (1) edge node {$d$} (2)
        (1) edge[bend right=40] node {$d$} (3)
        (1) edge[bend right=43] node {$d$} (4)
        (2) edge[loop above] node {$b,c$} (2)
        (2) edge node {$d$} (3)
        (3) edge[loop above] node {$b,c$} (3)
        (3) edge node {$d$} (4)
        (4) edge node {$b$} (5)
        (4) edge[bend right] node {$b$} (7)
        (5) edge node {$c$} (6)
        (6) edge node {$b$} (7)
        (7) edge node {$a,c$} (8)
        (7) edge[bend right=30] node[above] {$a$} (1);
    \end{tikzpicture}
    \caption{Automaton $A_1$; $n=8$ and $F=\{\frac{n}{2},n\}$.}
    \label{ex02_2}
  \end{figure}

  Note that the languages are disjoint since $A_0$ accepts strings ending with $b$,
  while $A_1$ accepts strings ending with $a,c$, or $d$.
  
  Consider the string
  \[
    \left[ \left(bd (bc)^{\frac{n}{4}}\right)^{\frac{n}{2}-2} bd(bc)^{\frac{n}{4}-1}ba \right]^{n-3} 
    \cdot 
    \left(bd (bc)^{\frac{n}{4}}\right)^{\frac{n}{2}-2} bd(bc)^{\frac{n}{4}-1} b c b \,.
  \]
  This string belongs to $L(A_0)$ and consists of $n-3$ parts each of length $\frac{n^2}{4}+\frac{n}{2}-2$, plus one part of length $\frac{n^2}{4}+\frac{n}{2}-1$. We can delete the last letters one by one, obtaining strings alternating between $L(A_1)$ and $L(A_0)$. Hence the length of this tower is $(n-2)\cdot (\frac{n^2}{4}+\frac{n}{2}-2) + 1$, which results in a tower of length $\Omega(n^3)$.

  To show that there is no infinite tower between the languages, we can use the techniques described in~\cite{icalp2013,mfcsPlaceRZ13}, or the algorithm presented in Section~\ref{secAlg}.
\qed\end{proof}

  As the last result of this paper, we prove an exponential lower bound with respect to the cardinality of the input alphabet.

  \begin{theorem}\label{thm05}
  \label{thm:exp}
    There exist two languages with no infinite tower having a finite tower of an exponential length with respect to the size of the alphabet.
  \end{theorem}

  \begin{proof}
    For every non-negative integer $m$,
 we define a pair of nondeterministic automata $A_m$ and $B_m$
 over the input alphabet $\Sigma_m=\{a_1,a_2,\ldots,a_m\}\cup\{b,c\}$ 
 with a tower of length $2^{m+2}$ 
 between $L(A_m)$ and $L(B_m)$, and such that
 there is no infinite tower between the two languages.

    The two-state automaton $A_m=(\{1,2\},\Sigma_m,\delta_m,1,\{2\})$ has self-loops under all symbols in state 1 and  a $b$-transition from state 1 to state 2.
The automaton 
 is shown in  Fig.~\ref{fig4} (left), 
 and it accepts all strings over $\Sigma_m$ ending with $b$.
    \begin{figure}[t]
      \centering
      \begin{tikzpicture}[->,>=stealth,shorten >=1pt,auto,node distance=1.8cm,
        state/.style={ellipse,minimum size=6mm,very thin,draw=black,initial text=}]
        \node[state,initial]    (1) {$1$};
        \node[state,accepting]  (2) [right of=1] {$2$};
        \path
          (1) edge[loop above] node {$\Sigma_m$} (1)
          (1) edge node {$b$} (2);
      \end{tikzpicture}
      \qquad
      \begin{tikzpicture}[->,>=stealth,shorten >=1pt,auto,node distance=1.8cm,
        state/.style={ellipse,minimum size=6mm,very thin,draw=black,initial text=}]
        \node[state,initial,accepting]    (1) {$p$};
        \node[state]                      (2) [right of=1] {$q$};
        \node[state,accepting]            (3) [right of=2] {$r$};
        \path
          (1) edge node {$b$} (2)
          (2) edge node {$c$} (3);
      \end{tikzpicture}
      \caption{The two-state NFA $A_m$, for $m\ge 0$ (left), and the automaton $B_0$ (right).}
      \label{fig4}
      \label{fig5}
    \end{figure}
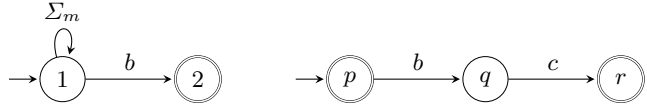
    
    The automata $B_m$ are constructed inductively as follows. The automaton $B_0=(\{p,q,r\},\{b,c\},\gamma_0,\{p\},\{p,r\})$ accepts the finite language $\{\varepsilon, bc\}$, and it is  shown in Fig.~\ref{fig5} (right).

 Assume that we have constructed the 
nondeterministic finite automaton 
$B_{m}=(Q_{m},\Sigma_{m},\gamma_{m},S_{m},\{p,r\})$. 
 We construct the nondeterministic  automaton
 $B_{m+1}=(Q_{m}\cup\{m+1\},\Sigma_{m} \cup\{a_{m+1}\},\gamma_{m+1},S_{m}\cup\{m+1\},\{p,r\})$
 by adding a new initial state 
$m+1$ to $Q_m$, 
and transitions on a fresh input symbol $a_{m+1}$. 
The transition function $\gamma_{m+1}$ extends $\gamma_{m}$ so that it defines self-loops under all letters of $\Sigma_{m}$ in the new state $m+1$,
    and adds the transitions on input
$a_{m+1}$  from state $m+1$ to all the states of $S_{m}$, that is, to all the initial states of $B_{m}$.
    The first two steps of the construction, that is, 
  automata $B_1$ and $B_2$,
are shown in Figs.~\ref{fig6} and~\ref{fig7}, respectively. Note that $L(B_{m})\subseteq L(B_{m+1})$ since all the initial states of $B_{m}$
  are initial in $B_{m+1}$ as well, and the set of final states is $\{p,r\}$
 in both automata.
    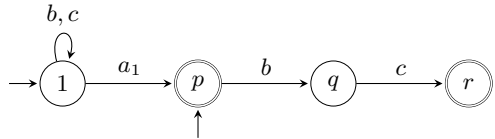
\begin{figure}[b]
      \centering
      \begin{tikzpicture}[->,>=stealth,shorten >=1pt,auto,node distance=1.8cm,
        state/.style={ellipse,minimum size=6mm,very thin,draw=black,initial text=}]
        \node[state,initial below,accepting]    (1) {$p$};
        \node[state]                      (2) [right of=1] {$q$};
        \node[state,accepting]            (3) [right of=2] {$r$};
        \node[state,initial]              (4) [left of=1] {$1$};
        \path
          (1) edge node {$b$} (2)
          (2) edge node {$c$} (3)
          (4) edge node[above] {$a_1$} (1)
          (4) edge[loop above] node {$b,c$} (4);
      \end{tikzpicture}
      \caption{Automaton $B_1$.}
      \label{fig6}
    \end{figure}
    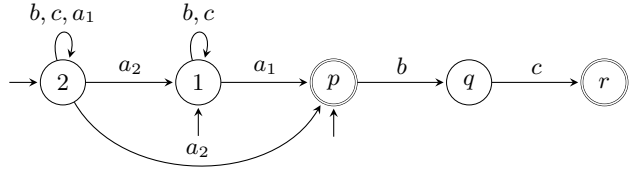
\begin{figure}[t]
      \centering
      \begin{tikzpicture}[->,>=stealth,shorten >=1pt,auto,node distance=1.8cm,
        state/.style={ellipse,minimum size=6mm,very thin,draw=black,initial text=}]
        \node[state,initial below,accepting]    (1) {$p$};
        \node[state]                      (2) [right of=1] {$q$};
        \node[state,accepting]            (3) [right of=2] {$r$};
        \node[state,initial below]              (4) [left of=1] {$1$};
        \node[state,initial]              (5) [left of=4] {$2$};
        \path
          (1) edge node {$b$} (2)
          (2) edge node {$c$} (3)
          (4) edge node[above] {$a_1$} (1)
          (4) edge[loop above] node {$b,c$} (4)
          (5) edge[loop above] node {$b,c,a_1$} (5)
          (5) edge[bend right=60] node {$a_2$} (1)
          (5) edge node {$a_2$} (4);
      \end{tikzpicture}
      \caption{Automaton $B_2$.}
      \label{fig7}
    \end{figure}
    
    By induction on $m$, 
  we show that there exists a tower
 between the languages $L(A_m)$ and $L(B_m)$
 of length $2^{m+2}$.
 More specifically, we prove that
 there exists a sequence $(w_i)_{i=1}^{2^{m+2}}$
 such that $w_i$ is a prefix of $w_{i+1}$
 and $|w_{i+1}|=|w_i|+1$ 
 for all $i=1,\ldots,2^{m+2}-1$,
 $w_1=\eps$, so $w_1\in L(B_m)$, and $w_{2^{m+2}}\in L(A_m)$.
 Thus, the tower is fully characterized by its longest string $w_{2^{m+2}}$. Moreover, by definition, the letter $b$ appears on all odd positions of $w_{2^{m+2}}$.
    
    If $m=0$, then such a tower is $\eps, b, bc, bcb$, and it is of length $2^2$. Assume that for some $m$, 
 we have a sequence of prefixes of length $2^{m+2}$ as required above,
 and such that the length of its longest string $wb$ is   $2^{m+2}-1$.
 Consider the automata $A_{m+1}$ and $B_{m+1}$ and the string 
    \[
      wba_{m+1}wb\,.
    \]
    The length of this string is $2^{(m+1)+2} - 1$, which results in $2^{(m+1)+2}$ prefixes. 
    
    By the assumption, every odd position is occupied by letter $b$,
 hence every prefix of an odd length belongs to $L(A_{m+1})$.
 It remains to show that all even-length prefixes belong to $L(B_{m+1})$.
 Let $x$ be such a prefix.
 If $x$ does not contain $a_{m+1}$, then it is a prefix of $wb$
 and belongs to $L(B_m)$ by the induction hypothesis.
 If $x=wba_{m+1}y$, then $B_{m+1}$ reads the string $wb$ in state $m+1$.   
 Then, on input $a_{m+1}$, it goes to an initial state of $B_m$.
 From this initial state,  the string $y$ is accepted as a prefix of $wb$
 by the induction hypothesis.
 Thus  $x$ is in $L(B_{m+1})$.
    
    To complete the proof,
 it remains to show that there is no infinite tower between the languages.
 We can either use the techniques described in~\cite{icalp2013,mfcsPlaceRZ13},
 or the algorithm presented in Section~\ref{secAlg}.
    However, to give a brief idea why it is so, we can give an inductive argument.  Since $L(B_0)$ is finite, there is no infinite tower between $L(A_0)$ and $L(B_0)$. Consider a tower between $L(A_{m+1})$ and $L(B_{m+1})$. 
 If~every string of the tower belonging to $L(B_{m+1})$
 is accepted from an initial state different from $m+1$, then 
    it is a tower between $L(A_{m})$ and $L(B_{m})$,
    so it is finite. Thus, if there exists an infinite tower, there also exists an infinite tower where all strings belonging to $L(B_{m+1})$ are accepted only from state $m+1$. However, every such string is of the form
 $(\{a_1,\ldots,a_{m}\}\cup\{b,c\})^* a_{m+1} y$,
 where the string $y$ is accepted from an initial state different from $m+1$.
 Cutting off the prefixes from $(\{a_1,\ldots,a_{m}\}\cup\{b,c\})^* a_{m+1}$ 
    results in an infinite tower between $L(A_{m})$ and $L(B_{m})$, which is a contradiction.
  \qed\end{proof}

\section{Conclusions}
  The definition of towers can be generalized from subsequences to basically any relation on strings, namely to prefixes, suffixes, etc. Notice that our lower-bound examples in Theorems~\ref{thm:quadratic}, \ref{thm:cubic}, and~\ref{thm:exp} are actually towers of prefixes, hence they give a lower bound on the length of towers of prefixes as well. 
  
  On the other hand, the upper-bound results cannot be directly used to prove the upper bounds for towers of prefixes. Although every tower of prefixes is also a tower of subsequences, the condition that there are no infinite towers is weaker for prefixes. The bound for subsequences therefore does not apply to languages that allow an infinite tower of subsequences but only finite towers of prefixes.
  
  Finally, note that the lower-bound results are based on nondeterminism. We are aware of a tower of subsequences (prefixes) showing the quadratic lower bound for deterministic automata. However, it is an open question whether a longer tower can be found or the upper bound is significantly different for deterministic automata.

\bibliographystyle{splncs03}
\bibliography{paper}

\begin{thebibliography}{10}
\providecommand{\url}[1]{\texttt{#1}}
\providecommand{\urlprefix}{URL }

\bibitem{Almeida-jpaa90}
Almeida, J.: Implicit operations on finite {J}-trivial semigroups and a
  conjecture of {I}.~{S}imon. Journal of Pure and Applied Algebra  69,
  205--218 (1990)

\bibitem{AlmeidaBook}
Almeida, J.: Finite semigroups and universal algebra, Series in Algebra,
  vol.~3. World Scientific (1995)

\bibitem{AlmeidaZ-ita97}
Almeida, J., Zeitoun, M.: The pseudovariety {J} is hyperdecidable. RAIRO --
  Theoretical Informatics and Applications  31(5),  457--482 (1997)

\bibitem{icalp2013}
Czerwi{\'n}ski, W., Martens, W., Masopust, T.: Efficient separability of
  regular languages by subsequences and suffixes. In: Proc. of ICALP. LNCS,
  vol. 7966, pp. 150--161. Springer (2013), full version available at
  http://arxiv.org/abs/1303.0966

\bibitem{pc2013}
Czerwi{\'n}ski, W., Martens, W., Masopust, T.: {P}ersonal communication. (2013)

\bibitem{higman}
Higman, G.: Ordering by divisibility in abstract algebras. Proceedings of the
  London Mathematical Society  s3-2(1),  326--336 (1952)

\bibitem{pkps2014}
P.~Karandikar, P.S.: On the state complexity of closures and interiors of
  regular languages with subwords. In: Proc. of DCFS (2014), to appear.
  Available at http://arxiv.org/abs/1406.0690

\bibitem{mfcsPlaceRZ13}
Place, T., van Rooijen, L., Zeitoun, M.: Separating regular languages by
  piecewise testable and unambiguous languages. In: Chatterjee, K., Sgall, J.
  (eds.) Proc. of MFCS. LNCS, vol. 8087, pp. 729--740. Springer (2013)

\bibitem{PlaceZ14}
Place, T., Zeitoun, M.: Separating regular languages with first-order logic.
  In: Proc. of CSL-LICS (2014), accepted. Available at
  http://arxiv.org/abs/1402.3277

\bibitem{Simon1972}
Simon, I.: Hierarchies of Events with Dot-Depth One. Ph.D. thesis, Dept. of
  Applied Analysis and Computer Science, University of Waterloo, Canada (1972)

\bibitem{Simon1975}
Simon, I.: Piecewise testable events. In: GI Conference on Automata Theory and
  Formal Languages. pp. 214--222. Springer (1975)

\bibitem{Stern-tcs85}
Stern, J.: Characterizations of some classes of regular events. Theoretical
  Computer Science  35,  17--42 (1985)

\end{thebibliography}
 
\end{document}